%% file: EA-Quaternions-arxiv.tex
\documentclass[a4paper,oneside,10pt,article]{amsart}

 \usepackage{amsaddr}

\usepackage{fullpage}

\usepackage{graphics}

\input{preamble2_macros.tex}

\newcommand{\R}{{\mathbb R}}

\newcommand{\Iset}{{\mathcal{I}}}

\newcommand{\qtf}{{\textsc{Q}}}
\newcommand{\iq}{{\text{i}}}
\newcommand{\jq}{{\text{j}}}
\newcommand{\kq}{{\text{k}}}
\newcommand{\ptf}{{\textsc{P}}}
\newcommand{\qtfa}{{\textsc{A}}}
\newcommand{\qtfb}{{\textsc{B}}}
\newcommand{\qtfr}{{\textsc{R}}}

\newcommand{\Qset}{{\mathcal{Q}}}
\newcommand{\wedgevec}{{\wedge_{\text{vec}} }}
\newcommand{\wedgequat}{{\wedge_{\Qset} }}

\newcommand{\pbf}{\mathbf{p}}
\newcommand{\qbf}{\mathbf{q}}

\newcommand{\fieldA}{\abf}

\newcommand{\fieldAi}{a}

\newcommand{\eg}{{e.g. }}

 \theoremstyle{theorem}
\newtheorem{proposition}{Proposition}
 
 \theoremstyle{theorem}
 \newtheorem*{remark}{Remark}
 
 \numberwithin{equation}{section}

 \theoremstyle{definition}
\newtheorem{property}{Property}[section]
 \theoremstyle{definition}
\newtheorem*{proof}{Proof}

\title{{Time-like definition of quaternions in 
exterior algebra}}

\author{Ivano Colombaro}
\address{Faculty of Engineering, Free University of Bozen-Bolzano
\\ {\small Piazza Universit\`{a} 5, 39100 Bolzano, ITALY.}	}
		\email{ivano.colombaro@unibz.it}
\thanks{Published in: \textbf{\textit{Ric. di Mat.} (2023). DOI:} \href{https://doi.org/10.1007/s11587-023-00810-z}{10.1007/s11587-023-00810-z}.}
\date{\today}

\begin{document}


\maketitle

\begin{abstract}
A formal description of quaternions by means of exterior calculus is presented. Considering a three-dimensional space-time characterized by three time-like coordinates, we have been able to consistently recover a suitable formulation of quaternions by means of the properties arising from exterior algebra and calculus.
As an application, it is also illustrated how rotations may be written in terms of quaternions, in accordance with definition provided in exterior algebra.
\end{abstract}



\section{Introduction}
Space-time exterior calculus serves as a valuable instrument that allows one to develop theories in mathematical physics involving general-graded multivectors in a framework of arbitrary time and space dimensions~\cite{colombaro2019introductionSpaceTimeExteriorCalculus}.
A generalized theory of electromagnetism, for instance, has been formulated through the properties of exterior calculus in~\cite{colombaro2020generalized}.
The introduction of exterior-algebraic variational methods permits one to find the dynamical equations of a system as the Euler-Lagrange equations~\cite{colombaro2021EulerLagrange} and to discover the conserved quantities such as the equivalent stress-energy-momentum tensor and the angular momentum.
In particular, for a suitably-defined Lagrangian density depending on multivector fields, it is possible to find a manifestly-symmetric form for the stress-energy-momentum tensor~\cite{martinez2021tensor, colombaro2021tensor-conf}, imposing the invariance of the action to infinitesimal space-time translations, while the invariance with respect to rotations allows one to recover the conservation law for the equivalent angular momentum~\cite{martinez2021angularmomentum}.

In this paper, we are going to define a suitable basis for quaternions in the context of exterior algebra, by considering a space-time characterized by three time dimensions.
Quaternions are mathematical objects first described in 1843 by W. R. Hamilton~\cite{hamilton1853LectureNotes}, who devised them and their operations to deal with three-dimensional problems in mechanics.
Later on, Clifford assimilated quaternions in geometric algebra, nowadays known as Clifford algebra~\cite{PerezGracia2020}. Indeed, quaternions include vector calculus in three dimensions and its characteristic operations of scalar and vector product and, historically, the algebraic description of quaternions might be considered as the origin of non-commutative algebra~\cite{Hazewinkel2004algebras}. 
Hamilton real quaternions, in particular, are the result of researching a way to extend the complex number to a higher-dimensional system~\cite{hamilton1844onquaternions}. 
Algebraically, quaternions are an example of a non-commutative division ring. Specifically, they satisfy the properties of a ring, as they form a closed set under addition and multiplication, as both addition and multiplication are associative and distributive, and admit an identity element. However, they do not satisfy the property of commutativity for multiplication, and for this reason quaternions identify a division ring but not a field.

In this work, we provide the definition and characterization of quaternions by means of the formalism of exterior calculus, presenting how the main operations appear in terms of the tools implemented in exterior algebra.
A different approach would consider Clifford algebra from the point of view of differential forms, in order to describe quaternions and the related quantities~\cite{salingaros19831}. However, for the purposes of this article, we are not explicitly adopting this latter approach, since a match is possible between differential forms and exterior calculus~\cite{colombaro2020generalized, colombaro2021EulerLagrange}.

Hereinafter, in Sect.~\ref{sec:EA} we briefly review the main concepts of the general $(k,n)$-dimensional construction of exterior algebra, suggesting some references for readers interested to deepen the topic. For the scope of this article, we will then limit our dimensions to three time-like coordinates.
Sect.~\ref{sec:quat-ea} is the core of the article where quaternions in exterior algebra are introduced and described.
In particular, in Sect.~\ref{sec:def-prob} we present the setup of the problem in the $(3,0)$ space-time and in Sect.~\ref{subsec:extprod+} we provide some necessary tools from exterior calculus which play a key role for the rest of the theory.
We deduce the exterior-algebraic formulation of quaternions in Sect.~\ref{ssec:tlquat} and, subsequently, we introduce the application to rotations of the exterior-algebraic quaternions in Sect.~\ref{sec:apps-ex}. Ultimately, we conclude the article with some final considerations in Sect.~\ref{sec:concl}.

\section{Basics in Exterior Calculus}\label{sec:EA}

In its most general formulation, exterior calculus takes into consideration a flat space-time $\Rb^{k+n}$ with $k$ temporal dimensions and $n$ spatial dimensions, called $(k,n)$ space-time. In this mathematical framework, we write the canonical basis $\displaystyle{\{\ebf_i\}_{i=0}^{k+n-1}}$, referring to indices from $0$ to $k-1$ as the time coordinates and to the following $n$ indices as the space coordinates.

Next, we proceed to introduce the fundamental elements of exterior algebra that are necessary for the purposes of this article, while readers seeking a more comprehensive understanding may find further exploration of the subject in~\cite{colombaro2019introductionSpaceTimeExteriorCalculus}.

\subsection{Exterior product, multivectors and main properties}\label{subsec:2.1}

A natural extension of the canonical vector basis can be done by means of the exterior (or wedge) product $\wedge$~\cite[p.~2]{winitzki2010}, and we define a grade-$r$ multivector basis  in a $(k,n)$ space-time, where $r\le k+n$, as in~\cite[Eq.~(5)]{colombaro2019introductionSpaceTimeExteriorCalculus},
\begin{equation}\label{eq:mv-basis}
\ebf_I = \ebf_{i_1} \wedge \ebf_{i_2} \wedge \dots \wedge \ebf_{i_r}, 
\end{equation}
where $I = (i_1,\dots,i_r)$ conventionally corresponds to the list of the ordered non-repeated indices $i_1, \dots , i_r$.
Using the basis and the notation adopted in~\eqref{eq:mv-basis}, we define a multivector or grade $r$ as
\begin{equation}\label{eq:multivector}
\fieldA(\xb) = \sum_{I\in\Iset_r} \fieldAi_I(\xb)\ebf_I,
\end{equation}
where $\Iset_r$ is the set of all ordered lists $I$ with non-repeated $r$ indices, for $r=0, 1, \dots, k+n-2, k+n-1$. Henceforth, we refer to the grade of the multivector field also as $\gr(\fieldA)=r$ and the length of the list $I$ is denoted as $\len{I}$ .

Let us now briefly summarize the main operations and properties among multivector fields, defined according to~\eqref{eq:multivector}. First, the dot product
$\cdot$ between two arbitrary grade-$r$ multivector basis $\ebf_I$ and $\ebf_J$ is defined as
\begin{equation}\label{eq:dot_multi}
	\ebf_I\cdot\ebf_J = \Delta_{IJ} = \Delta_{i_1 j_1}\Delta_{i_2 j_2}\dots\Delta_{i_r j_r},
\end{equation}
where $I$ and $J$ are the ordered lists of $r$ indices, $I = (i_1,i_2,\dots,i_r)$ and $J = (j_1,j_2,\dots,j_r)$, and $\Delta_{IJ}$ is the generalized metric tensor, where $\Delta_{ij} = 0$ if $i\neq j$, $\Delta_{ii} = -1$ for $i=0,\dots, k-1$ and $\Delta_{ii} = +1$, for $i=k, \dots, k+n-1$. 

We have introduced the exterior product between vector basis in~\eqref{eq:mv-basis} and we now define the same operation between two multivector basis $\ebf_I$ and $\ebf_J$, having grades respectively $r = \len{I}$ and $r' = \len{J}$, namely
\begin{equation} \label{eq:ext-prod-def}
	\ebf_I\wedge\ebf_J = \sigma(I,J)\ebf_{I+J},
\end{equation}
where $\sigma(I,J)$ is the signature of the permutation sorting the elements of this concatenated list of $\len{I}+\len{J}$ indices, and $I+J$ represents the resulting sorted list. 

We then proceed to describe two generalizations of the dot product, namely the left and right interior products. Let us consider two basis multivectors $\ebf_I$ and $\ebf_J$, such that $I$ is a subset of $J$, thus $\len{I}\le \len{J}$ and $I\subseteq J$. Then the left interior product, denoted as $\lintprod$, is defined as
\begin{equation}	\label{eq:left-int-prod}
	\ebf_I \lintprod \ebf_J = \Delta_{II}\sigma(J\setminus I,I)\ebf_{J\setminus I},
\end{equation}
and the right interior product, denoted as $\rintprod$, is defined as 
\begin{equation}\label{eq:right-int-prod}
	\ebf_J \rintprod \ebf_I = \Delta_{II}\sigma(I,J\setminus I)\ebf_{J\setminus I}, 
\end{equation}
where the vector $\ebf_{J\setminus I}$ has grade $\len{J}-\len{I}$ and it is composed by the elements of $J$ not in common with $I$. Left and right interior products are connected to each other by means of the relation~\cite{colombaro2019introductionSpaceTimeExteriorCalculus}
\begin{equation}
\ebf_I \lintprod \ebf_J = \ebf_J \rintprod \ebf_I (-1)^{\len{I}(\len{J} - \len{I})},
\end{equation}
and both these operations coincide with the dot product in~\eqref{eq:dot_multi} in the case that the two lists have the same grade $\len{I}= \len{J}$, namely 
\begin{equation}
\ebf_I \lintprod \ebf_J = \ebf_I \rintprod \ebf_J = \ebf_I \cdot \ebf_J , \qquad \len{I}= \len{J}.
\end{equation}
Furthermore, for a grade-$r$ multivector basis $\ebf_I$, we may denote what we call as the Grassmann or Hodge complement~\cite{Frankel2012geomphys} and we write it~\cite[Sect.~2.1]{colombaro2019introductionSpaceTimeExteriorCalculus}
\begin{equation} \label{eq:hodge-transf}
	\ebf_I^\hodge = \Delta_{I,I}\sigma(I,I^c)\ebf_{I^c},
\end{equation}
where $I^c$ is the complement list composed of all the labels but the elements of $I$ and identifying the grade-$(k+n-r)$ dual multivector.
Similarly, we express the inverse Hodge transformation as
\begin{equation} \label{eq:hodge-inv-transf}
	\ebf_I^\hodgeinv = \Delta_{I^c,I^c}\sigma(I^c,I)\ebf_{I^c} .
\end{equation}

\section{Quaternions in Exterior Algebra}\label{sec:quat-ea}

\subsection{Setup of the problem} \label{sec:def-prob}

Let us consider a space-time composed of three time-like coordinates. Conventionally, we name these coordinates $\ebf_\iq$, $\ebf_\jq$ and $\ebf_\kq$ and, as a consequence, the generalized metric tensor admits non-vanishing terms $\Delta_{\iq\iq}=\Delta_{\jq\jq}=\Delta_{\kq\kq}=-1$.
In this framework, from~\eqref{eq:hodge-transf} and~\eqref{eq:hodge-inv-transf} we can easily recover the following expressions for the Hodge and inverse Hodge dual of vectors, respectively
\begin{gather}
\ebf_\iq^\hodge = - \sigma(\iq, \jq\kq) \ebf_{\jq\kq}
\\
\ebf_\iq^\hodgeinv = \sigma(\jq\kq,\iq) \ebf_{\jq\kq}	,\label{eq:vec-invhodge}
\end{gather}
and the same relations for bivectors, specifically
\begin{gather}
\ebf_{\iq\jq}^\hodge = \sigma(\iq\jq,\kq)\ebf_\kq \label{eq:bivec-hodge}
\\
\ebf_{\iq\jq}^\hodgeinv = - \sigma(\kq,\iq\jq)\ebf_\kq , \label{eq:bivec->vec}
\end{gather}
pointing out that $\ebf_\iq^\hodgeinv =- \ebf_\iq^\hodge $ and $\ebf_{\iq\jq}^\hodgeinv = -\ebf_{\iq\jq}^\hodge$,
due to the nature of the $(3,0)$-dimensional space-time.

We then define entities denoted as $\qtf$, constructed from four real numbers, $a$, $b$, $c$ and $d$. One of them, generally $a$, is multiplied by a scalar basis, which might be specifically written as $\ebf_\emptyset$ but it is omitted without any confusion. The three remaining numbers $b$, $c$ and $d$, are instead multiplied by the three basis elements of our $(3,0)$ space-time. Hence, we write $\qtf$ as
\begin{equation}\label{eq:quat-gen-def}
\qtf = a + b\ebf_\iq + c\ebf_\jq + d\ebf_\kq ,
\end{equation}
and, hereinafter, we will name the full set of objects characterized in~\eqref{eq:quat-gen-def} with the symbol $\Qset$ and we may simply refer by the element $\qtf$ as $\qtf \in \Qset$.

\subsection{Exterior product and its variants}\label{subsec:extprod+}
Exterior product, also known as wedge product, may increase the degree of the operation we are computing~\cite{colombaro2019introductionSpaceTimeExteriorCalculus}, as we have noticed in~\eqref{eq:mv-basis} and~\eqref{eq:ext-prod-def}.

\begin{property}\label{prop:biv}
We assume that the basis elements we are taking into consideration do not exceed grade one. Thus, in case we find bivectors, we write them by means of their Hodge complement, using~\eqref{eq:bivec->vec}.
In this way, objects as outlined in~\eqref{eq:quat-gen-def} define a closed group, since the basic operation of exterior product among elements like~\eqref{eq:quat-gen-def} generate objects with the same structure, as will be further explored in the subsequent sections of this paper.
Bivectors appear as
\begin{equation}
\ebf_{pq} =\ebf_p \wedgevec \ebf_q = \sigma(pq,r) \ebf_r , \label{eq:bivec-vec}
\end{equation}
where $p,\, q ,\, r = \iq , \jq, \kq$ assume the values of the three time-like coordinates of our space-time. In these specific cases, we are using the symbol $\wedgevec$, in order to identify this particular operation and to avoid confusion. 
\end{property}

Property~\ref{prop:biv} has the effect of identifying the two-dimensional surface described by a bivector by means of its orientation in the three dimensional space that we are considering.
In fact, in our $(3,0)$ space-time as in $\R^3$, there is an equivalence between the operations of wedge product and cross product~\cite{Lounesto2001clifford}, which we identify with the symbol $\wedgevec$ in order to avoid confusion.
Thus, the product $\wedgevec$ has the role of ``vectorizing'' the bivector and, in other words, we identify the wedge product $\wedgevec$ with the standard cross product $\times$ between two vectors in $\R^3$.

To conclude this section, we introduce a generalized product among elements of $\Qset$, represented as $\wedgequat$. Given two objects $\qtf_1, \qtf_2 \in \Qset$, we define
\begin{equation}\label{eq:wedgeQ}
\qtf_1 \wedgequat \qtf_2 = \qtf_1 \cdot \qtf_2 + \qtf_1 \wedgevec \qtf_2 ,
\end{equation}
where the first term on the right-hand-side is the dot product producing the scalar part of the result, while the second term on the right-hand-side determines the vector part by means of the product defined in~\eqref{eq:bivec-vec}.

\subsection{Definition of Time-Like Quaternions}\label{ssec:tlquat}

Under the conditions that we have introduced in the previous section, we can identify $\qtf\in\Qset$ in~\eqref{eq:quat-gen-def} as the exterior-algebraic representation for the quaternions in {Clifford} algebra~\cite[Sect.~1.1]{Vaz2016IntroCliffordAlg}. The configuration presented in~\ref{subsec:2.1} can already be connected with the setup described in~\cite[Sect.~1]{ata2018adifferentpolar}, where the negative (resp.~positive) signature is related to time-like (resp.~space-like) coordinates. 
Thus, we find a time-like representation of quaternions in exterior calculus, which is also equivalent to the Clifford-algebraic time-like representation in Minkowski 3-space~\cite{inoguchi1998timelikesurfaces, ozdemir2006rotations}.

In the first instance, employing the product defined at~\eqref{eq:wedgeQ}, we can easily deduce the property 
\begin{equation}
\ebf_\iq \wedgequat \ebf_\iq = \ebf_\jq \wedgequat \ebf_\jq = \ebf_\kq \wedgequat \ebf_\kq  = \ebf_\iq \wedgequat \ebf_\jq \wedgequat \ebf_\kq = -1 ,
\end{equation}
having its correspondence in the classical definition of the algebra of quaternions~\cite{clifford1878applications}.

Furthermore, we can also prove the equivalence of the Hamilton product between two elements, which for us is simply identified with the wedge product between two objects. Let us consider two objects $\ptf, \qtf \in \Qset$, satisfying the properties above mentioned, defined respectively
\begin{gather}
\ptf = p_0 + p_1\ebf_\iq + p_2\ebf_\jq + p_3\ebf_\kq , \label{eqP}
\\
\qtf = q_0 + q_1\ebf_\iq + q_2\ebf_\jq + q_3\ebf_\kq . \label{eqQ}
\end{gather} 
The $\Qset$-product between them is
\begin{equation}\label{eq:hprod}
\ptf \wedgequat \qtf = \alpha + \beta \ebf_\iq + \gamma \ebf_\jq + \delta \ebf_\kq,
\end{equation}
where
$\alpha =  (p_0 q_0 - p_1 q_1 -p_2 q_2 - p_3 q_3)$, $\beta = (p_0 q_1 + q_0 p_1 + p_2 q_3 - p_3 q_2)$, $\gamma= (p_0 q_2 + q_0 p_2 + q_1 p_3 - p_1 q_3)$ and $\delta = (p_0 q_3 + q_0 p_3 + p_1 q_2 - q_1 p_2)$, recovering the standard expression known in literature as Hamilton product~\cite[Sect.~5.4]{Kuipers1999quaternions}.
We may also write the quaternions in~\eqref{eqP} and~\eqref{eqQ} separating explicitly the scalar from the vectorial part, namely
\begin{gather}
\ptf = p_0 + \pbf , \label{eqP2}
\\
\qtf = q_0 + \qbf , \label{eqQ2}
\end{gather} 
where $\pbf=(p_1, p_2, p_3)$ and $\qbf=(q_1, q_2, q_3)$ are three-component vectors in terms of the time-like basis triplet $\{\ebf_\iq, \ebf_\jq, \ebf_\kq\}$.
The Hamilton product in~\eqref{eq:hprod} is then expressed in terms of~\eqref{eqP2} and~\eqref{eqQ2} as
\begin{equation} \label{eq:PhamQ}
\ptf \wedgequat \qtf = p_0 q_0 + \pbf \cdot \qbf + q_0 \pbf + p_0 \qbf + \pbf \wedgevec \qbf ,
\end{equation}
or, recalling that the vectors $\pbf$ and $\qbf$ have a three-dimensional time-like vector basis, we write in an explicit way the scalar product as
\begin{equation}
\ptf \wedgequat \qtf = p_0 q_0 - p_1 q_1 - p_2 q_2 - p_3 q_3 + q_0 \pbf + p_0 \qbf + \pbf \times \qbf ,
\end{equation}
where we have expressed the vector wedge product $\wedgevec$ as the standard cross product, without any confusion.

As for the classical definition of the quaternions, we might also define the conjugate of an object $\qtf\in \Qset$, represented by $\qtf^*$, namely
\begin{equation}\label{eq:ccNOhodge}
\qtf^* = a - b\ebf_\iq - c\ebf_\jq - d\ebf_\kq ,
\end{equation}
and identified by the opposite sign in the vector elements.
\begin{proposition}\label{prop:ccgp}
For $\ptf, \qtf \in \Qset$, it holds that
\begin{equation}\label{eq:ccgp}
(\ptf \wedgequat \qtf)^* = \qtf^* \wedgequat \ptf^* ,
\end{equation}
where the generalized product $\wedgequat$ is defined in~\eqref{eq:wedgeQ} and the conjugate changes the sign of the time-vector elements according to~\eqref{eq:ccNOhodge}.
\end{proposition}
\begin{proof}
In order to prove the relation in~\eqref{eq:ccgp}, 
we consider $\ptf$ and $\qtf$ as written in~\eqref{eqP2} and~\eqref{eqQ2}, respectively. On the left hand side of~\eqref{eq:ccgp} we recognize the conjugate of the expression in~\eqref{eq:PhamQ}, namely
\begin{equation}
(\ptf \wedgequat \qtf)^* = p_0 q_0 + \pbf \cdot \qbf - q_0 \pbf - p_0 \qbf - \pbf \times \qbf ,
\end{equation}
by simply changing the sign of the vector terms.
On the right hand side, we have instead to evaluate 
\begin{equation}
\qtf^* \wedgequat \ptf^* = (q_0 - \qbf) \wedgequat (p_0 - \pbf) ,
\end{equation}
which results, after a few computations, 
\begin{equation}
\qtf^* \wedgequat \ptf^* = p_0 q_0 + \pbf \cdot \qbf - q_0 \pbf - p_0 \qbf - \qbf \times \pbf.
\end{equation}
Since in the last term holds that $- \qbf \times \pbf = \pbf \times \qbf$, then Proposition~\ref{prop:ccgp} is proved.
\end{proof}

We are also able to write the squared norm of $\qtf$ in accordance with the standard definition
\begin{equation}\label{eq:Qnorm}
\vert \qtf \vert^2 
= \qtf\wedgequat\qtf^*  = q_0^2 + q_1^2 + q_2^2 + q_3^2 ,
\end{equation}
and we tipically refer to unit quaternions in the case when the norm is unitary. In addition, we can evaluate the norm of the generalized product in~\eqref{eq:wedgeQ}, which reads
\begin{equation}\label{eq:norm-gen-prod}
\vert \ptf \wedgequat \qtf \vert^2 = \vert \ptf \vert^2 \vert \qtf \vert^2 .
\end{equation}
\begin{proof}
It is possible to demonstrate this latter expression in a few steps. First, we may apply the definition of the squared norm as in~\eqref{eq:Qnorm} and we get
\begin{equation}
\vert \ptf \wedgequat \qtf \vert^2 = (\ptf \wedgequat \qtf ) \wedgequat ( \ptf \wedgequat \qtf )^*.
\end{equation}
Subsequently, using~\eqref{eq:ccgp}, we are able to write 
\begin{equation}
\vert \ptf \wedgequat \qtf \vert^2 = \ptf \wedgequat \qtf \wedgequat \qtf^* \wedgequat \ptf^* = \ptf \wedgequat \ptf^* \vert \qtf \vert^2 = \vert \ptf \vert^2 \vert \qtf \vert^2 ,
\end{equation}
likewise~\eqref{eq:norm-gen-prod}, employing again the definition of the norm in~\eqref{eq:Qnorm}.
\end{proof}

We investigate the inverse of a quaternion $\qtf \in \Qset$, which we are going to name $\qtf^{-1}$ adopting the common notation, satisfying the relations
\begin{equation}\label{eq:inv-def-id}
\qtf^{-1} \wedgequat \qtf = \mathbf{1} =\qtf \wedgequat \qtf^{-1},
\end{equation}
where $\mathbf{1}$ is the identity operator. Taking into account the equality on the left of~\eqref{eq:inv-def-id}, we multiply both terms on the right for the factor $\wedgequat \qtf^{*}$ and we get
\begin{equation}
\qtf^{-1} \wedgequat \qtf \wedgequat \qtf^{*} = \qtf^{*}.
\end{equation}
Recognizing the norm $\vert \qtf \vert ^2 = \qtf \wedgequat \qtf^{*} $, it is then easy to derive
\begin{equation}\label{eq:def-inv-quat}
\qtf^{-1}= \frac{\qtf^{*}}{\vert \qtf \vert ^2 }.
\end{equation}
We could have found the same result by considering the identity on the right of~\eqref{eq:inv-def-id} and multiplying for the factor $\qtf^{*} \wedgequat$ on the left, since $\qtf \wedgequat \qtf^{*} = \qtf^* \wedgequat \qtf$.
\begin{remark}\label{remark:unitrot}
Note that in the case of unit quaternion, namely $\vert \qtf \vert ^2 = 1$, we have that $\qtf^{-1}=\qtf^{*}$, in analogy with rotation matrices.
\end{remark}

\section{Quaternions and Rotations}\label{sec:apps-ex}
One of the most important application concerning quaternions is surely dealing with rotations. In fact, it is particularly remarkable to use quaternions for three-dimensional rotations in order to reduce memory consumption and compute matrices faster, \eg in computer graphics~\cite[Sect.~11-2]{Hearn1997ComputerGraphics}.

In order to represent vector rotations with quaternions, we first need a unit quaternion, which can be written in the form~\cite[Chap.~X]{Brand1947vectorandtensor}
\begin{equation}\label{eq:quaternion-polar}
\qtf = \cos \varphi + \sin \varphi \, \ebf ,
\end{equation}
where the vector part could be made explicit as $\displaystyle \ebf = \frac{q_1 \ebf_\iq + q_2 \ebf_\jq + q_3 \ebf_\kq}{\sqrt{q_1^2+q_2^2+q_3^2}}$.
Second, given two vectors $\abf=(a_1, a_2, a_3)$ and $\bbf=(b_1, b_2, b_3)$, we notice that they can be expressed as \textit{pure quaternions}, which are a particular kind of quaternion with vanishing scalar term, defined by $\qtfa = 0 + \abf$ and $\qtfb = 0 + \bbf$.
Then, the unit quaternion in~\eqref{eq:quaternion-polar} can be written as a quotient of these two latter pure quaternions, namely
\begin{equation}\label{eq:quotientquat}
 \cos \varphi + \sin \varphi \, \ebf  = \qtfb \wedgequat \qtfa^{-1} ,
\end{equation}
where the inverse of $\qtfa$ is defined, in accordance with~\eqref{eq:def-inv-quat}, as
\begin{equation}
\qtfa^{-1} = -\frac{\abf}{\vert \abf \vert^2} .
\end{equation}
To employ the expression $\qtfb \wedgequat \qtfa^{-1}$ in~\eqref{eq:quotientquat}, it is necessary to satisfy a few requirements on the two vectors $\abf$ and $\bbf$: they should have the same length $\vert \abf \vert = \vert \bbf \vert$ and $\varphi$ should be the angle between them. Furthermore, vectors $\abf$ and $\bbf$ have both to be perpendicular to $\ebf$ so that the three $\{\abf, \,\bbf,\,\ebf\}$ identify a right-handed set.
Once satisfying the required premises, it is easy to verify that the right hand side of~\eqref{eq:quotientquat} becomes
\begin{equation}\label{eq:BA-scal-vec}
\qtfb \wedgequat \qtfa^{-1} = \frac{1}{\vert \abf \vert^2} \left( -\abf\cdot\bbf+ \abf\times \bbf\right),
\end{equation}
or, alternatively,
\begin{equation}
\qtfb \wedgequat \qtfa^{-1} = \frac{1}{\vert \abf \vert^2} \left( a_1 b_1 + a_2 b_2 + a_3 b_3 + \abf\times \bbf\right).
\end{equation}
An additional intuitive explanation to~\eqref{eq:quotientquat} might also be deduced from~\eqref{eq:BA-scal-vec}. The scalar part $-\abf\cdot\bbf$ is indeed $\cos\varphi$, due to the condition $\vert \abf \vert = \vert \bbf \vert$ and to the definition of the angle $\varphi$. Moreover, since $\ebf$ is perpendicular to both $\abf$ and $\bbf$, it is then straightforward to prove that $\abf\times\bbf = \sin\varphi \, \ebf$.

We can simply motivate the analogy with rotations and, as a consequence, a geometrical interpretation, by multiplying, in generalized sense $\wedgequat$, both sides of~\eqref{eq:quotientquat} by the quantity $\qtfa$ and taking advantage of~\eqref{eq:inv-def-id}, so that we obtain the relation
\begin{equation}
\qtfb = \qtf \wedgequat \qtfa , 
\end{equation}
and substituting the expression for the quaternion in~\eqref{eq:quaternion-polar}, we get
\begin{equation}\label{eq:rotationQA}
\qtf \wedgequat \qtfa = \abf \cos \varphi + \ebf \times \abf \sin \varphi .
\end{equation}
Looking at~\eqref{eq:rotationQA}, we recognize the equivalent expression for a rotation by the angle $\varphi$ of $\abf$ with respect to $\ebf$. In this specific case, the rotation is confined to a two-dimensional plane which is perpendicular to $\ebf$~\cite{pujol2012Hamilton}.

In general, a rotation through an angle $\varphi$ about a unit vector $\ubf=(u_1, u_2, u_3)$, expressed in the three-dimensional time-like basis, is written by means of the unit quaternion~\cite[Chap.~X]{Brand1947vectorandtensor}
\begin{equation}\label{eq:Qforvarphi}
\qtf = \cos\frac{\varphi}{2} + \ubf\sin\frac{\varphi}{2}.
\end{equation}
Then, we consider a pure quaternion $\qtfr = 0+\mathbf{r}$, generated by the vector $\mathbf{r}$, and we write its transformation as
\begin{equation}\label{eq:Rrot}
\qtfr' = \qtf \wedgequat \qtfr \wedgequat \qtf^{-1} ,
\end{equation}
with $\qtfr' = 0+\mathbf{r}'$ and $\mathbf{r}'$ corresponding to the rotated three-dimensional vector.
To clarify this latter concept, let us investigate an example. We consider~\eqref{eq:Qforvarphi} where $\ubf=\ebf_\iq$ and  $\mathbf{r} = \ebf_\jq$.
As a consequence,~\eqref{eq:Rrot} may be written as follows
\begin{equation}
\mathbf{r}' = \left( \cos\frac{\varphi}{2} + \ebf_\iq\sin\frac{\varphi}{2}\right)\wedgequat \ebf_\jq \wedgequat \left( \cos\frac{\varphi}{2} - \ebf_\iq\sin\frac{\varphi}{2}\right),
\end{equation}
which can be simplified, after a few computations,
\begin{align}
\mathbf{r}' &= \left( \cos\frac{\varphi}{2} + \ebf_\iq\sin\frac{\varphi}{2}\right)\wedgequat \left( \ebf_\jq \cos\frac{\varphi}{2} + \ebf_\kq\sin\frac{\varphi}{2}\right)
\\
&= \ebf_\jq \left(  \cos^2\frac{\varphi}{2} - \sin^2\frac{\varphi}{2}\right) + 2\, \ebf_\kq \sin\frac{\varphi}{2}\cos\frac{\varphi}{2}
\\
&= \ebf_\jq\, \cos\varphi + \ebf_\kq \, \sin\varphi	.	\label{eq:asdfg}
\end{align}
Focus our attention on the final expression in~\eqref{eq:asdfg}, resulting
\begin{equation}
\mathbf{r}'  = \ebf_\jq' =\ebf_\jq\, \cos\varphi + \ebf_\kq \, \sin\varphi ,
\end{equation}
it is remarkable that it corresponds to the vector obtained by rotating, in positive sense, $\ebf_\jq$ about $\ebf_\iq$ through an angle $\varphi$.

\section{Conclusions}\label{sec:concl}

In this paper, we have provided a description of quaternions from the perspective of exterior algebra.
We have seen how the exterior product may be adapted in the three-dimensional time-like reference frame, in order to determine a natural characterization of some mathematical objects, which are perfectly compatible with the standard definition of quaternions.
We have also discussed the application to rotations, in its more basic representation of rotation matrices. Beyond the purposes of this work but concerning rotational problems, some interesting insight could be found considering the Lagrangian approach and the conservation of angular momentum, as in~\cite{martinez2021angularmomentum}.
Working on this topic, an interesting approach from the algebraic point of view, could be inspired by~\cite[Ch. 6]{Vaz2016IntroCliffordAlg}, describing the expressions found for angular momentum and spin in exterior calculus in view of the results obtained.


Another possible extension of the quaternionic multiplications might be done in regards of generalized Maxwell equations in exterior algebra, which have been found following two different approaches in~\cite{colombaro2020generalized} and~\cite{colombaro2021EulerLagrange}.
Considering the results of this article, multivectorial electromagnetism in exterior calculus could  be extended as shown by E. K\"{a}hler in~\cite{kahler1937maxwell}. Moreover, inspired by N. Salingaros~\cite{salingaros1981electromagnetism}, other interesting observations might arise, as a formal equivalence between $r$-vectors and $p$-forms, following the steps as described in~\cite{salingaros1979aadd} and computing the suitable operations.
%
%
%

Furthermore, future perspectives can surely deal with other applications of quaternions, whose investigation touches different and varied topics and appears in a large variety of physical problems.


\section*{Acknowledgements}
The author acknowledge the anonymous reviewer for the constructive suggestions which have helped to improve the manuscript.
Also, the author is deeply grateful to Josep Font-Segura for many helpful comments and discussions.
The work of the author has been carried out in the framework of the activities of the Italian National Group of Mathematical Physics [Gruppo Nazionale per la Fisica Matematica (GNFM), Istituto Nazionale di Alta Matematica (INdAM)]. 
Moreover, the work of the author is partially supported by INdAM-GNFM Young Researchers Project 2023, CUP\_E53C22001930001.

\bibliographystyle{unsrt}	
\bibliography{physics-rdm.bib}

\end{document}

%% file: preamble2_macros.tex
\usepackage{amsmath}
\usepackage{amsthm}
\usepackage{amssymb}
\usepackage{amstext}
\usepackage{graphicx}
\usepackage{stmaryrd}

\usepackage{hyperref} 

\usepackage{color}
\usepackage[applemac]{inputenc}

\usepackage{pict2e}

\usepackage{xparse}

\makeatletter
\DeclareRobustCommand{\lintprod}{%
  \mathbin{\mathpalette\int@prod{(0,0)(0.8,0)(0.8,0.6)}}%
}
\DeclareRobustCommand{\rintprod}{%
  \mathbin{\mathpalette\int@prod{(0.1,0.6)(0.1,0)(0.9,0)}}}

\newcommand{\int@prod}[2]{%
  \begingroup
  \sbox\z@{$\m@th#1+$}%
  \setlength\unitlength{\wd\z@}%
  \linethickness{0.09\unitlength}%
  \begin{picture}(1,1)
  \roundcap
  \polyline#2
  \end{picture}%
  \endgroup
}
\makeatother

\usepackage[retainorgcmds]{IEEEtrantools}

\usepackage{tikz}

\newcommand{\Rb}{{\mathbf R}}

\newcommand{\ebf}{{\mathbf e}}

\newcommand{\xb}{{\mathbf x}}

\newcommand{\ubf}{{\mathbf u}}

\DeclareDocumentCommand \spinori { o } {%
  \IfNoValueTF {#1} {%
    \sigma%
  }{%
    \sigma_{#1}%
  }%
}

\DeclareDocumentCommand \indG { o o } {%
  \IfNoValueTF {#1} {%
    \iota%
  }{%
    \IfNoValueTF {#2} {%
      \iota_{#1}%
    }{%
      \iota_{#1}(#2)%
    }
  }%
}

\DeclareDocumentCommand \indGc { o o } {%
  \IfNoValueTF {#1} {%
    \bar\iota%
  }{%
    \IfNoValueTF {#2} {%
      \bar\iota_{#1}%
    }{%
      \bar\iota_{#1}(#2)%
    }
  }%
}

\DeclareDocumentCommand \indGinv { o o } {%
  \IfNoValueTF {#1} {%
    \bar\iota%
  }{%
    \IfNoValueTF {#2} {%
      \iota^{-1}_{#1}%
    }{%
      \iota^{-1}_{#1}(#2)%
    }
  }%
}

\DeclareDocumentCommand \coefG { o o } {%
  \IfNoValueTF {#1} {%
    \gamma%
  }{%
    \IfNoValueTF {#2} {%
      \gamma_{#1}%
    }{%
      \gamma_{#1}(#2)%
    }
  }%
}

\newcommand{\abf}{\mathbf{a}}
\newcommand{\bbf}{\mathbf{b}}

\DeclareDocumentCommand \act { o } {%
  \IfNoValueTF {#1} {%
    \mathcal S%
  }{%
    \mathcal S_{\text{#1}}%
  }%
}

\DeclareDocumentCommand \ld { o } {%
  \IfNoValueTF {#1} {%
    \mathcal L%
  }{%
    \mathcal L_{\text{#1}}%
  }%
}

\newcommand{\xbpert}{\pmb{\varepsilon}}

\newcommand{\Gb}{\mathbf{G}}
\makeatletter
\def\trMs{\@ifnextchar[{\@with}{\@without}}
\def\@with[#1]{\Gb_{\xbpert}^{#1}}
\def\@without{\Gb_{\xbpert}^s}
\makeatother

\DeclareMathOperator{\gr}{gr}

\newcommand{\len}[1]{\lvert#1\rvert}

\newcommand{\hodge}{{\scriptscriptstyle\mathcal{H}}}
\newcommand{\hodgeinv}{{\scriptscriptstyle\mathcal{H}^{-1}}}

\newtheoremstyle{question}
  {\topsep}   
  {\topsep}   
  {\upshape}  
  {0pt}       
  {\itshape}  
  {.}         
  {5pt plus 1pt minus 1pt} 
  {\thmname{#1} \thesection.\thmnumber{\itshape#2}\thmnote{(#3)}} 

\makeatletter
    \def\@endtheorem{\hfill$\P$\endtrivlist\@endpefalse }
\makeatother
\theoremstyle{question}

\makeatletter
\let\save@mathaccent\mathaccent
\newcommand*\if@single[3]{%
  \setbox0\hbox{${\mathaccent"0362{#1}}^H$}%
  \setbox2\hbox{${\mathaccent"0362{\kern0pt#1}}^H$}%
  \ifdim\ht0=\ht2 #3\else #2\fi
  }
\newcommand*\rel@kern[1]{\kern#1\dimexpr\macc@kerna}
\newcommand*\widebar[1]{\@ifnextchar^{{\wide@bar{#1}{0}}}{\wide@bar{#1}{1}}}
\newcommand*\wide@bar[2]{\if@single{#1}{\wide@bar@{#1}{#2}{1}}{\wide@bar@{#1}{#2}{2}}}
\newcommand*\wide@bar@[3]{%
  \begingroup
  \def\mathaccent##1##2{%
    \let\mathaccent\save@mathaccent
    \if#32 \let\macc@nucleus\first@char \fi
    \setbox\z@\hbox{$\macc@style{\macc@nucleus}_{}$}%
    \setbox\tw@\hbox{$\macc@style{\macc@nucleus}{}_{}$}%
    \dimen@\wd\tw@
    \advance\dimen@-\wd\z@
    \divide\dimen@ 3
    \@tempdima\wd\tw@
    \advance\@tempdima-\scriptspace
    \divide\@tempdima 10
    \advance\dimen@-\@tempdima
    \ifdim\dimen@>\z@ \dimen@0pt\fi
    \rel@kern{0.6}\kern-\dimen@
    \if#31
      \overline{\rel@kern{-0.6}\kern\dimen@\macc@nucleus\rel@kern{0.4}\kern\dimen@}%
      \advance\dimen@0.4\dimexpr\macc@kerna
      \let\final@kern#2%
      \ifdim\dimen@<\z@ \let\final@kern1\fi
      \if\final@kern1 \kern-\dimen@\fi
    \else
      \overline{\rel@kern{-0.6}\kern\dimen@#1}%
    \fi
  }%
  \macc@depth\@ne
  \let\math@bgroup\@empty \let\math@egroup\macc@set@skewchar
  \mathsurround\z@ \frozen@everymath{\mathgroup\macc@group\relax}%
  \macc@set@skewchar\relax
  \let\mathaccentV\macc@nested@a
  \if#31
    \macc@nested@a\relax111{#1}%
  \else
    \def\gobble@till@marker##1\endmarker{}%
    \futurelet\first@char\gobble@till@marker#1\endmarker
    \ifcat\noexpand\first@char A\else
      \def\first@char{}%
    \fi
    \macc@nested@a\relax111{\first@char}%
  \fi
  \endgroup
}
\makeatother

\newcounter{aside}
   {\refstepcounter{aside}
   \begin{adjustwidth}{\parindent}{-1.5\parindent}
	\small \textbf{Aside \thechapter.\theaside~#1}
	}
{\par\end{adjustwidth}}